\documentclass[sn-mathphsy-sum,envcountsame,envcountsect]{llncs}
\usepackage{versions}
\usepackage[utf8]{inputenc}
\usepackage[T1]{fontenc}

\usepackage[mathlines]{lineno}

\let\doendproof\endproof
\renewcommand\endproof{~\hfill$\qed$\doendproof}

\usepackage[disable,textsize=scriptsize]{todonotes}%

\usepackage{bm} 
\usepackage{amssymb} 
\usepackage{url}
\usepackage{listings}
\usetikzlibrary{arrows.meta,arrows}
\usepackage{mathtools}
\usepackage{cleveref}

\newcommand{\shuffle}{\hskip.3ex plus .1ex minus .1ex%
  \mathchoice{{\scriptstyle\sqcup\hskip-.16em\sqcup}}%
	     {{\scriptstyle\sqcup\hskip-.16em\sqcup}}%
	     {{\scriptscriptstyle\sqcup\hskip-.16em\sqcup}}%
	     {{\scriptscriptstyle\sqcup\hskip-.16em\sqcup}}%
  \hskip.3ex plus .1ex minus .1ex}
\newcommand{\Nat}{\ensuremath{\mathbb{N}}}%
\newcommand{\ta}{\mathtt{A}}
\newcommand{\tb}{\mathtt{B}}
\newcommand{\tc}{\mathtt{C}}
\newcommand{\calA}{\mathcal A}%
\newcommand{\calB}{\mathcal B}%
\newcommand{\dwc}{\mathop{\downarrow}\nolimits}%
\newcommand{\alphabet}{\operatorname{\textit{alph}}}
\newcommand{\Id}{\mathit{Id}}
\newcommand{\subword}{\preccurlyeq}%
\newcommand{\nsubword}{\not\preccurlyeq}%
\newcommand{\subsim}{\lesssim}
\newcommand{\eqby}[1]{\stackrel{\!#1\!}{=}}

\newcommand{\leqby}[1]{\stackrel{\!#1\!}{\leq}}
\newcommand{\eqdef}{\stackrel{\mbox{\begin{scriptsize}def\end{scriptsize}}}{=}}
\newcommand{\equivdef}{\stackrel{\mbox{\begin{scriptsize}def\end{scriptsize}}}{\Leftrightarrow}}
\newcommand{\tuple}[1]{\langle #1 \rangle}

\title{
       On the piecewise complexity of words
}

\author{
        Ph.\ Schnoebelen
\and
        I.\ Vialard
}

\institute{
    Laboratoire M\'ethodes Formelles, Univ.\ Paris-Saclay, France
}

\begin{document}

\maketitle

\begin{abstract}
The piecewise complexity $h(u)$ of a word is the minimal length of
subwords needed to exactly characterise $u$. Its piecewise minimality
index $\rho(u)$ is the smallest length $k$ such that $u$ is minimal
among its order-$k$ class $[u]_k$ in Simon's congruence.

We initiate a study of these two descriptive complexity measures. Among
other results we provide efficient algorithms for computing $h(u)$ and
$\rho(u)$ for a given word $u$.

\end{abstract}

\section*{Foreword for editor and reviewers}
This submission is an extended version of some results presented at
the SOFSEM 2024 conference. The conference paper was titled ``On the
piecewise complexity of words \emph{and periodic words}'' [emphasis
added] and was signed by four authors: that is, the two authors of this
here submission, joined by M.~Praveen and J.~Veron,
see~\cite{PSVV-sofsem2024}. This situation warrants explanations.

The work and results in this submission address the piecewise
complexity of words and is due to Ph.~Schnoebelen joined recently by
I.~Vialard. The SOFSEM conference paper presented an abridged and
incomplete version of these results, together with some extra
developments for periodic words that were obtained in recent
collaborations with M.~Praveen and J.~Veron. This resulted in a
conference paper addressing two related but separate issues. In
retrospect we realised this was not the best redactional choice and
decided that, when it comes to final publication in a peer-reviewed
journal where more details (including full proofs) are expected, the
wise decision is to make two separate articles. This submission is the
first of the two, signed by its two authors, and the second
submission, when it is ready, will be signed by all four authors.

\newpage

\section{Introduction}
\label{sec-intro}

For two words $u$ and $v$, we write $u\subword v$ when $u$ is a
\emph{subword}, i.e., a subsequence, of $v$. For example
$\mathtt{SIMON}\subword \mathtt{STIMULATION}$ while
$\mathtt{IMRE} \not \subword \mathtt{NOISEMAKER}$.
Subwords and subsequences play
a prominent role in many areas of combinatorics and computer science. Our personal
motivations come from descriptive complexity and the possibility of
characterising words and languages via some short witnessing subwords.

In the early 1970s, and with similar motivations, Imre Simon
introduced \emph{piecewise-testable} (PT) languages in his PhD thesis
(see~\cite{simon72,simon75,sakarovitch83}): a language $L$ is PT if
there is a \emph{finite} set of words $F$
such that the membership of a word $u$ in $L$ depends only on which words from
$F$ are subwords of $u$. Since then, PT languages  have	 played an
important role in the algebraic and logical theory of first-order
definable languages, see~\cite{pin86,DGK-ijfcs08,klima2011,place2022} and the
references therein.  They also constitute an important class of simple
regular languages with applications in learning
theory~\cite{kontorovich2008}, databases~\cite{bojanczyk2012b},
linguistics~\cite{rogers2013}, etc. The concept of piecewise-testability has
been extended to variant notions of ``subwords''~\cite{zetzsche2018},
to trees~\cite{bojanczyk2012b}, pictures~\cite{matz98}, infinite
words~\cite{perrin2004,carton2018b}, or any
combinatorial well-quasi-order~\cite{goubault2016}.

When a PT language $L$ can be characterised via a finite $F$ where all
words have length at most $k$, we say that $L$ is piecewise-testable
\emph{of height $k$}, or $k$-PT.  Equivalently, $L$ is $k$-PT if it is
closed under
Simon's congruence of order $k$, denoted  $\sim_k$, that relates words
having the same subwords of length
at most $k$. The \emph{piecewise complexity} of $L$, denoted $h(L)$ (for
``height''), is the smallest $k$ such that $L$ is $k$-PT. It coincides
with the minimum number of variables needed in any $\calB\Sigma_1$ formula that
defines $L$~\cite{thomas82,DGK-ijfcs08}.

The piecewise complexity of languages was studied by
Karandikar and Schnoebelen in~\cite{KS-lmcs2019} where it is a central
tool for establishing elementary upper bounds for the complexity of
the $\mathsf{FO}^2$ fragment of the logic of
subwords.

In this article, we focus on the piecewise complexity of
\emph{individual words}. For $u\in A^*$, we write $h(u)$ for
$h(\{u\})$, i.e., the smallest $k$ such that $[u]_k=\{u\}$, where $[u]_k$
is the equivalence class of $u$ w.r.t.\ $\sim_k$.  We also introduce a
new measure, $\rho(u)$, defined as the smallest $k$ such that $u$ is
minimal (with respect to the  subword ordering) in $[u]_k$.

Let us explain the idea more concretely.  A word can be defined via
subword constraints. For example,  $\mathtt{A B B A}$ is the only word
$u$ over the alphabet $A=\{\ta,\tb\}$ that satisfies the following
constraints:
\begin{gather}
\label{eq-ABBA-constraints}
\mathtt{A A}\subword u \:,
\quad
\mathtt{A A A}\not\subword u \:,
\quad
\mathtt{B B}\subword u \:,
\quad
\mathtt{B B B}\not\subword u \:,
\quad
\mathtt{A A B}\not\subword u \:,
\quad
\mathtt{B A A}\not\subword u \:.
\end{gather}
Indeed, the first four constraints require that $u$ contains exactly two
$\ta$'s and two $\tb$'s, while the
last two constraints require that any $\tb$ in $u$ sits between the two $\ta$'s. Since
all the constraints use subwords of length at most
3, we see that $h(\ta\tb\tb\ta)$ is at most 3.\footnote{We'll see
later how to prove that $\ta\tb\tb\ta$ cannot be defined with
constraints using shorter subwords so that its piecewise complexity is
exactly 3.} Note that we are not concerned with the number of
constraints, as we rather focus
on the length of the subwords in use. Given some $v$ different from
$u=\ta\tb\tb\ta$, we now know that there is a short witness of this
difference: a word of length at most 3 that is subword of only one
among $u$ and $v$. This witness, called a \emph{distinguisher}, depends on $v$ ---actually it can
be chosen among the constraints listed
in~\eqref{eq-ABBA-constraints}--- but the length bound comes from $u$
only.

Such small piecewise complexities often occur even
for long words: e.g., both $h(\mathtt{ABRACADABRA})$ and
$h(\mathtt{THE\,FULL\,WORKS\,OF\,WILLIAM\,SHAKESPEARE})$ equal $4$.
\\

We have two main motivations for this study. Firstly,
it appeared in~\cite{KS-lmcs2019} that bounding $h(L)$ for a PT
language $L$ relies heavily on knowing $h(u)$ for specific words $u$
in and out of $L$. For example, the piecewise complexity of a
\emph{finite} language $L$ is always  $\max_{u\in L}h(u)$, and
the tightness of many upper bounds in~\cite{KS-lmcs2019} relies on
identifying a family of long words with small piecewise
complexity.
Secondly,
the piecewise
complexity of words raises challenging combinatorial and algorithmic
questions. As a starting point, one would like a practical and efficient
algorithm that computes $h(u)$ for any given $u$.

\subsubsection*{Our contribution.}

Along $h(u)$, we introduce a new measure, $\rho(u)$, the \emph{piecewise
minimality index} of $u$, and initiate an investigation of the
combinatorial and algorithmic properties of both measures.  The new
measure $\rho(u)$ is closely related to $h(u)$ but is easier to
compute. Our main results are (1) theoretical results connecting $h$
and $\rho$ and bounding their values in contexts involving
concatenation, and (2) efficient algorithms for computing $h(u)$ and
$\rho(u)$.

\subsubsection*{Related work.}

The existing literature on piecewise complexity mostly considers $h(L)$ for
$L$ a PT-language, and provide general bounds (see,
e.g.,~\cite{KS-lmcs2019,HS-ipl2019}). We are not aware of any
practical algorithm computing $h(L)$ for $L$ a PT-language given,
e.g., via a deterministic finite-state automaton $\calA$, and it is
known that, for any given $k$, deciding whether $h(L(\calA))\leq k$ is in
$\textsf{coNP}$~\cite[Prop.~1]{masopust2015}. For individual words,
\cite[Prop.~3.1]{KS-lmcs2019}
states that $h(u)$ is polynomial-time computable but this is just a
direct reading of \Cref{eq-hu-delta-u-u1au2} from
\Cref{section-for-eq-hu-delta-u-u1au2} below, without any attempt at an
efficient algorithm.

A related notion is the \emph{piecewise distance} between words,
denoted $\delta$,
as defined by Sakarovitch and Simon already in~\cite{sakarovitch83}.
The literature here is a bit richer: Simon
claimed an $O(|uv|)$-time algorithm computing $\delta(u,v)$
in~\cite{simon2003} but the algorithm was never published.  The
problem received renewed attention
recently~\cite{fleischer2018,barker2020} and Gawrychowski et
al.\ eventually produced a linear-time
algorithm~\cite{gawrychowski2021}.  As we show below
(\Cref{prop-hu-delta-u-u1au2,prop-rho-charac}), there is a clear
connection between $\delta$ on one hand, and $h$ and $\rho$ on the
other hand.  However, these connections do not, or not
immediately, lead to efficient algorithms for computing $h$ or $\rho$.

\subsubsection*{Outline of the article.}
After recalling the necessary background in \Cref{sec-basics}, we
define the new measures $h(u)$ and $\rho(u)$ in \Cref{sec-h-rho} and
prove some first elementary properties like monotonicity and
convexity.  In \Cref{sec-algo-h-and-rho} we give efficient
algorithms for computing
$h(u)$ and $\rho(u)$.  Finally, in \Cref{sec-binary} we focus on binary
words and give a specialised, more efficient algorithm, for $h$ and $\rho$.

\section{Words, subwords and Simon's congruence}
\label{sec-basics}

We consider finite words $u,v,\ldots$ over a finite non-empty alphabet $A$.
We write $|u|$ for the length of a word $u$ and
$|u|_a$ for the number of occurrences of a letter  $a$	in $u$.
The
empty word is denoted with $\epsilon$.
For a word $u=a_1a_2\cdots a_n$ of length $|u|=n$ and two positions $i,j$
such that $0\leq
i\leq j\leq n$, we write $u(i,j)$ for the
factor $a_{i+1}a_{i+2}\cdots a_j$. Note that $u(0,|u|)=u$, that
$u(i_1,i_2)\cdot u(i_2,i_3)=u(i_1,i_3)$, and that $|u(i,j)|=j-i$. We
write $u(i)$ as shorthand for $u(i-1,i)$, i.e., $a_i$, the $i$-th
letter of $u$. With $\alphabet(u)$ we denote the set of letters that
occur in $u$. We often abuse notation and write ``$a\in u$'' instead
of ``$a\in\alphabet(u)$'' to say that a letter $a$ occurs in a word
$u$.

We say that $u=a_1\cdots a_n$ is a \emph{subword} of $v$, written
$u\subword v$, if $v$ can be factored under the form $v=v_0a_1v_1 a_2
\cdots v_{n-1} a_n v_n$ where the $v_i$'s can be any words (and can be
empty).	 We write $\dwc u$ for the \emph{downward closure} of $u$,
i.e., the set of all subwords of $u$: e.g.,
$\dwc\mathtt{ABAA} = \{\epsilon, \mathtt{A}, \mathtt{B}, \mathtt{A A},
\mathtt{A B}, \mathtt{B A}, \mathtt{A A A}, \mathtt{A B A}, \mathtt{B
A A}, \mathtt{A B A A}\}$.

Factors are a special case of subwords: $u$ is a \emph{factor} of $v$
if $v=v'u v''$ for some $v',v''$. Furthermore, when $v=v'u v''$ we say
that $u$ is a \emph{prefix} of $v$ when $v'=\epsilon$, and is a
\emph{suffix} of $v$ when $v''=\epsilon$.

When $u\neq v$, a word $s$ is a \emph{distinguisher} (or a
\emph{separator}) if $s$ is subword of exactly one word among $u$ and
$v$~\cite{simon72}. Observe that any two distinct words admit a distinguisher.

For $k
\in\Nat$ we write $A^{\leq k}$ for the set of words over $A$ that have
length at most $k$, and	 for any words $u,v \in A^*$, we let $ u \sim_k v
\equivdef \dwc{u}\cap A^{\leq k}=\dwc{v}\cap A^{\leq k} $.  In other
words, $u\sim_k v$ if $u$ and $v$ have the same subwords of length at
most $k$.  For example $\mathtt{ABAB}\sim_1 \mathtt{AABB}$ (both words
use the same letters, which are their subwords of length 1) but $\mathtt{ABAB}\not\sim_2 \mathtt{AABB}$
($\mathtt{BA}$ is a subword of $\mathtt{ABAB}$, not of
$\mathtt{AABB}$).
The equivalence $\sim_k$, introduced in~\cite{simon72,simon75}, is called \emph{Simon's congruence of order $k$}.
Note that $u\sim_0 v$ for any $u,v$, and $u\sim_k u$ for any
$k$. Finally, $u\sim_{k+1} v$ implies $u\sim_k v$ for any $k$, and
there is a refinement hierarchy
${\sim_0}\supseteq{\sim_1}\supseteq{\sim_2}\supseteq\cdots$ with
$\bigcap_{k\in\Nat}\sim_k=\Id_{A^*}$.
We write $[u]_k$ for the equivalence class of $u\in A^*$ under
$\sim_k$.  Note that each $\sim_k$, for $k=0,1,2,\ldots$, has finite
index~\cite{simon75,sakarovitch83,KKS-ipl2015}.

We further let
$u \subsim_k v \equivdef u\sim_k v \land u\subword v$.
Note that $\subsim_k$ is stronger than $\sim_k$. Both relations are
(pre)congruences: $u\sim_k v$ and $u'\sim_k v'$ imply $uu'\sim_k vv'$,
while $u\subsim_k v$ and $u'\subsim_k v'$ imply $uu'\subsim_k vv'$.

The following properties will be useful:
\begin{lemma}
\label{lem-useful}
For all $u,v,v',w \in A^*$ and $a,b \in A$:
\begin{enumerate}

\item
\label{it-convex}
If  $u \sim_k v$ and  $u\subword w\subword v$ then $u\subsim_k
w\subsim_k v$;

\item
\label{it-carac-richn}
When $k>0$, $u\sim_k u v$ if, and only if, there exists a
factorization $u=u_1 u_2\cdots u_k$ such that $\alphabet(u_1)
\supseteq \alphabet(u_2) \supseteq \cdots \supseteq \alphabet(u_k) \supseteq
\alphabet(v)$;

\item
\label{it-diff-let}
If $u a v\sim_k u b v'$ and $a\neq b$ then
$u bav\sim_k u b v'$ or $u ab v'\sim_k u a v$ (or both);

\item
\label{it-upperbound}
If  $u \sim_k v$ then there exists $w \in
A^*$ such that $u \subsim_k w$ and $v \subsim_k w$;

\item
\label{it-shorter}
If $u\sim_k v$ and $|u|<|v|$ then there exists some $v'$ with
$|v'|=|u|$ and such that $u\sim_k v'\subword v$;

\item
\label{it-pumping}
If $u v \sim_k u a v$ then
$u v \sim_k u a^m v$ for all $m\in\Nat$;

\item
\label{it-sing-inf}
Every equivalence class of $\sim_k$ is a singleton or is infinite.

\end{enumerate}
\end{lemma}
\begin{proof}
(\ref{it-convex}) is by combining
$\dwc{u}\subseteq \dwc{w}\subseteq \dwc{v}$ with the definition of $\sim_k$;
(\ref{it-carac-richn}--\ref{it-upperbound}) are Lemmas 3, 5, and 6 from~\cite{simon75};
(\ref{it-shorter}) is an immediate consequence of Theorem~4 from
\cite[p.~91]{simon72}, showing that all minimal (w.r.t.\ $\subword$) words in $[u]_k$ have
the same length --- see also
\cite[Theorem~6.2.9]{sakarovitch83} or \cite{fleischer2018} ---;
(\ref{it-pumping}) is in the proof of Corollary~2.8 from~\cite{sakarovitch83};
(\ref{it-sing-inf}) follows from (\ref{it-convex}), (\ref{it-upperbound}) and (\ref{it-pumping}).
\end{proof}

The fundamental tools for reasoning about piecewise complexity were
developed in Simon's thesis~\cite{simon72}.\footnote{Actually
$\delta$ does not appear \emph{per se} in~\cite{simon72}
or~\cite{simon75}. There, it is used implicitly in the definition of
$r$ and $\ell$ via Equations~\eqref{eq-def-r}
and~\eqref{eq-def-l}. Only in~\cite{sakarovitch83} does $\delta$
appear explicitly, both as a convenient notation and as a first-class
tool.  }
First, there is the concept of \emph{subword ``distance''}\footnote{In fact,
$\delta(u,v)$ is a measure of \emph{similarity} and not of difference,
between $u$ and $v$. A proper piecewise-based distance would be, e.g.,
$d(u,v)\eqdef 2^{-\delta(u,v)}$~\cite{sakarovitch83}.}
$\delta(u,v)\in\Nat\cup\{\infty\}$,  defined for any $u,v\in A^*$, via
\begin{align}
\label{eq-delta-distinguisher}
\delta(u,v)
&\eqdef \max\{k~|~u\sim_k v\}
\\
\label{eq-delta-distinguisher-2}
&=\begin{cases}
	\infty &\text{if $u=v$,}
	\\
	|s|-1 &\text{if $u\neq v$ and $s$ is a shortest distinguisher.}
\end{cases}
\end{align}
Derived notions are the left and right distances~\cite[p72]{simon72}, defined for any
$u,t\in A^*$, via
\begin{align}
\label{eq-def-r}
r(u,t)&\eqdef \delta(u,u t) = \max\{k~|~u\sim_k u t\} \:,
\\
\label{eq-def-l}
\ell(t,u) &\eqdef \delta(t u,u) = \max\{k~|~t u\sim_k u\} \:.
\end{align}

Clearly $r$ and $\ell$ are mirror notions. One usually
proves properties of $r$ only, and (often implicitly) deduce
symmetrical conclusions for $\ell$ by the mirror reasoning.

\begin{lemma}[{\cite[Lemma~6.2.13]{sakarovitch83}}]
\label{lem-uv-uav}
For any words $u,v\in A^*$ and letter $a\in A$
\begin{gather}
\label{eq-uv-uav}
\delta(u v,u a v)=\delta(u,u a)+\delta(a v,v)=r(u,a)+\ell(a,v)
\:.
\end{gather}
\end{lemma}

\section{The piecewise complexity of words}
\label{sec-h-rho}

In this section, we define the complexity measures $h(u)$ and
$\rho(u)$, give characterisations in terms of the side
distance functions  $r$ and $\ell$, compare the two measures and establish some first results on the measures of concatenations.

\subsection{Defining words via their subwords}
\label{ssec-h-def}

The piecewise complexity of piecewise-testable (PT) languages was defined
in~\cite{KS-csl2016,KS-lmcs2019}. Formally, for a language $L$ over
$A$, $h(L)$ is the smallest index $k$ such that $L$ is
$\sim_k$-saturated, i.e., closed under $\sim_k$ (and we let $h(L)=\infty$
when $L$ is not PT).  When $L$ is a
singleton $\{u\}$ with $u\in A^*$, this becomes $h\bigl(\{u\}\bigr)
\eqdef \min\{n~|~ \forall v:u\sim_n v\implies u=v\}$: we write this
more simply as $h(u)$ and call it the \emph{piecewise complexity} of
$u$. For example, $h(a^n)=n+1$ when $a$ is any letter.

Our first result expresses $h$ in terms of $\delta$.
\label{section-for-eq-hu-delta-u-u1au2}
\begin{proposition}[After \protect{\cite[Prop.~3.1]{KS-lmcs2019}}]
For any $u\in A^*$,
\label{prop-hu-delta-u-u1au2}
\begin{align}
\label{eq-hu-delta-u-u1au2}
h(u)
&= \max_{\substack{u=u_1u_2 \\ a\in A}} \delta(u,u_1a u_2)+1
\\
\label{eq-hu-delta-u1-u1a-u2-au2}
\tag{H}
&= \max_{\substack{u=u_1u_2 \\ a\in A}}r(u_1,a)+\ell(a,u_2)+1
\:.
\end{align}
\end{proposition}

\begin{proof}
We only prove \eqref{eq-hu-delta-u-u1au2} since
\eqref{eq-hu-delta-u1-u1a-u2-au2}  is just a
rewording based on \Cref{eq-uv-uav}.  \\
$(\geq)$: By definition, $h(u)>\delta(u,v)$ for any $v\neq u$, and in
particular for any $v$ of the form $v=u_1 a u_2$.  \\
$(\leq)$: Let $k=h(u)-1$. By definition of $h$, there exists some
$v\neq u$ with $u\sim_k v$. By
\Cref{lem-useful}~(\ref{it-upperbound}), we can further assume
$u\subsim_k v$, and by \Cref{lem-useful}~(\ref{it-convex}), we can
even assume that $|v|=|u|+1$, i.e., $v = u_1 a u_2$ for some $a\in A$
and some factorization $u=u_1u_2$.  Now $\delta(u,v)=\delta(u,u_1 a
u_2)\geq k$ since $u\sim_k v$.	Thus $h(u)=k+1\leq \delta(u,u_1 a
u_2)+1$ for this particular choice of $u_1,u_2$ and $a$.
\end{proof}

\subsection{Reduced words and the minimality index}
\label{ssec-rho-def}

\begin{definition}[{\cite[p.~70]{simon72}}]
Let $m>0$, a word $u \in A^{*}$ is \emph{$m$-reduced} if
$u\not\sim_m u'$ for all strict subwords $u'$ of $u$.
\end{definition}
In other words, $u$ is $m$-reduced when it is a $\subword$-minimal word in
$[u]_{m}$.
This leads to a new piecewise-based measure for words, that we call
the \emph{minimality index}:
\begin{equation}
\rho(u) \eqdef \min \{ m ~|~ u \text{ is $m$-reduced} \}
\:.
\end{equation}
In particular, $\rho(\epsilon)=0$ and more generally
$\rho(a^n)=n$ for any letter $a\in A$.

\begin{lemma}[{\cite[p.~72]{simon72}}]
\label{lem-reduced}
A non-empty word $u$ is $m$-reduced iff $r(u_1,a)+\ell(a,u_2) < m$ for all
factorizations $u=u_1 a u_2$ with $a\in A$ and $u_1,u_2\in A^*$.
\end{lemma}
\begin{proof}
Assume, by way of contradiction, that $r(u_1,a)+\ell(a,u_2)\geq m$ for
some factorization $u=u_1 a u_2$.
\Cref{lem-uv-uav} then gives $\delta(u,u_1u_2)\geq
m$, i.e.,  $u\sim_m u_1 u_2$. Finally, $u$ is not
minimal in $[u]_m$.
\end{proof}

This has an immediate corollary:
\begin{proposition} For any non-empty word $u\in A^+$
\label{prop-rho-charac}
\begin{align}
\label{eq-m-via-r&l}
\tag{P}
\rho(u) = \max_{\substack{u=v_1av_2\\ a\in A}} r(v_1,a) + \ell(a,v_2) + 1
\:.
\end{align}
\end{proposition}

Note the difference between
\Cref{eq-hu-delta-u1-u1a-u2-au2,eq-m-via-r&l}: $h(u)$ can be computed
by looking at all ways one would insert a letter $a$ inside $u$ while,
for $\rho(u)$, one is looking at all ways one could remove some letter
from $u$.

The similarity between \Cref{eq-hu-delta-u1-u1a-u2-au2,eq-m-via-r&l}
is what prompted the inclusion of $\rho(u)$ in our investigation. We
should, however, note an important difference between the two
measures: it is very natural to see $h$ as a measure for languages
(indeed, this is how it was first defined~\cite{KS-lmcs2019}),
while $\rho$ only makes sense when applied to individual words.

\begin{remark}
Formally, the definitions of $h(u)$ and $\rho(u)$ depend on the
underlying alphabet $A$ and we should more precisely write
$h_A(u)$ and $\rho_A(u)$. However, the only situation where the values
of $h_A(u)$
and $\rho_A(u)$
is influenced by $A$ is when $u=\epsilon$ for which we have
$h_\emptyset(\epsilon)=0$ while $h_A(\epsilon)=1$ for non-empty $A$
(no differences for $\rho)$. This has potential impact on several of
our equations and for this reason we assumed that $A$ is non-empty
throughout the article. With this assumption, no confusion is created
when writing $h(u)$ and $\rho(u)$ without any subscript.
\qed
\end{remark}

\subsection{Fundamental properties of side distances}

The characterisations given in
\Cref{prop-hu-delta-u-u1au2,prop-rho-charac} suggest that
computing $h(u)$ and $\rho(u)$ reduces to computing the $r$ and $\ell$
side distance functions on prefixes and suffixes of $u$. This will be
confirmed in \Cref{sec-algo-h-and-rho}.

For this reason, we now prove some useful combinatorial results on $r$
and $\ell$. They will be essential for proving more general properties of $h$
and $\rho$ in the rest of this section, and in the analysis of algorithms
in the next section.

The following  lemma provides a recursive way of computing
$r(u,t)$.
\begin{lemma}[{\cite[p.~71--72]{simon72}}]
For any $u,t\in A^*$ and $a\in A$:
\label{lem-sld-comp}
\begin{align}
\label{eq-sld-t}
\tag{R1}
r(u,t) &= \min \bigl\{r(u, a) ~|~ a\in\alphabet(t)\bigr\} \:,
\\
\label{eq-sld-0}
\tag{R2}
r(u,a) &= 0 \text{ if $a$ does not occur in $u$} \:,
\\
\label{eq-sld-u'}
\tag{R3}
r(u,a) &= 1 + r(u',a u'') \text{ if $u=u' a u''$ with
$a\not\in\alphabet(u'')$} \:.
\end{align}
\end{lemma}
\begin{proof}
\eqref{eq-sld-t}: Since $r(u,\epsilon)=\delta(u,u)=\infty=\min\emptyset$,
the statement holds when $t=\epsilon$.	We now assume $t\neq
\epsilon$ and $r(u,t)\in\Nat$.

Pick $a$ occurring in $t$ and let $k=1+r(u,a)$. By
\Cref{eq-delta-distinguisher-2}, there exists $s$ of length $k$ such
that $s\nsubword u$ and $s\subword u a$. Hence $s\subword u t$, and $s$
is a distinguisher for $u$ and $u t$. With Eq.~\eqref{eq-delta-distinguisher-2} we deduce $r(u,
t)\leq |s|-1=r(u,a)$. Since this
holds for any $a$ in $t$, we conclude
$r(u,t)\leq\min_{a\in\alphabet(t)}r(u,a)$.

For the ``$\geq$'' direction, let $k=1+r(u, t)$ and pick a
distinguisher $s$ of length $k$ such that $s\nsubword u$ and
$s\subword u t$.  Write $s=s_1a s_2$ with $s_1$ the longest prefix of
$s$ that is a subword of $u$, $a$ the first letter after $s_1$, and
$s_2$ the rest of $s$. Now $a s_2\subword t$ so that
$a\in\alphabet(t)$. Since $s_1a\nsubword u$, we deduce $r(u, a)\leq
|s_1a|-1\leq |s|=1+r(u, t)$. We have then found some $a\in\alphabet(t)$
with $r(u,t)\geq r(u,a)$.  \\

\eqref{eq-sld-0}:
If $a$ does not occur in $u$, it is a distinguisher with $u a$
hence $r(u,a)=0$ by \Cref{eq-delta-distinguisher-2}.
\\

\eqref{eq-sld-u'}:
Assume $|u|_a>0$ and write $u=u' a u''$ with $|u''|_a=0$.

Let $k=1+r(u',a u'')$. By \Cref{eq-delta-distinguisher-2} there exists
a distinguisher $s$ of length $k$ with $s\nsubword u'$ and $s\subword
u' a u''$, further entailing $s a\nsubword u' a u''=u$ and $s a\subword
u' a u'' a=u a$, i.e., $s a$ is a distinguisher for $u$ and $u a$.  We
deduce $r(u,a)\leq |s a|-1=k=1+r(u',a u'')$, proving the ``$\leq$''
direction of \eqref{eq-sld-u'}.

For the other direction, let $k=r(u,a)$. Then there is a distinguisher
$s$ of length $k+1$ with $s\nsubword u$ and $s\subword
u a$. Necessarily $s$ is some $ta$, with $t\subword u$. From
$ta\nsubword u$, we deduce $t\nsubword u'$ and then $t$ distinguishes
between $u'$ and $u=u' a u''$. Then $r(u',a u'')\leq |t|-1\leq k-1$,
i.e., $1+r(u',a u'')\leq k=r(u,a)$.
\end{proof}

\begin{corollary}
\label{coro-monot-subalphabet}
\label{coro-dr-a-notin-v}
For any $u,v,t,t'\in A^*$ and $a\in A$
\begin{align}
\label{eq-monot-subalphabet}
\alphabet(t)\subseteq \alphabet(t')
&\implies
r(u,t)\geq r(u,t') \text{ and }
\ell(t,u)\geq \ell(t',u)
\:,
\\
\label{eq-dr-a-notin-v}
a\not\in\alphabet(v)
&\implies r(u v,a) \leq r(u,a) \text{ and }
\ell(a,v u) \leq \ell(a,u)
\:.
\end{align}
\end{corollary}
\begin{proof}
We only prove the claims on $r$ since the claim on $\ell$ can be
derived by mirroring.
The implication in~\eqref{eq-monot-subalphabet} is direct
from~\eqref{eq-sld-t}.
For the proof of \eqref{eq-dr-a-notin-v}, the assumption is that $a$
does not occur in $v$ and we consider two cases:

(i)
if  $a$ does not occur in $u$, we have
$r(u v,a) \eqby{\eqref{eq-sld-0}} 0
\eqby{\eqref{eq-sld-0}}r(u,a)$.

(ii) if $a$ occurs in $u$, we write $u=u' a u''$ with $a$ not
occurring in $u''$, so that $r(u v,a)
\eqby{\eqref{eq-sld-u'}}
1+r(u',a u''v)
\leqby{\eqref{eq-monot-subalphabet}}
1+r(u',a u'')
\eqby{\eqref{eq-sld-u'}}
r(u,a)$.
\end{proof}

Compared with \Cref{lem-sld-comp} used by Simon, our next Lemma provides  an alternation
version that leads to simple algorithms in later sections, and that is
very handy in proofs by induction on $u$.
\begin{lemma}
\label{lem-new-r-rule*}
For all $u\in \Sigma^*$ and any letters $a,b\in \Sigma$:
\begin{equation}
\label{eq-new-r-rules*}
\tag{R4}
\begin{aligned}
r(\epsilon,a) & = 0
\\
r(ub,a)	      & =
	      \begin{cases}
		r(u,a) + 1		   &\text{if $a=b,$}
		\\
		\min\bigl(r(u,b)+1,\: r(u,a)\bigr) &\text{if $a\neq b.$}
	      \end{cases}
\end{aligned}
\end{equation}
\end{lemma}
\begin{proof}
First note that $r(\epsilon,a)=0$ and $r(u a,a)=r(u,a)+1$ are special
cases of Equations \eqref{eq-sld-0} and \eqref{eq-sld-u'},
respectively. Now there only remains to consider $r(u b,a)$ assuming
$a\neq b$.

\noindent
($\geq$): Following \Cref{eq-delta-distinguisher-2,eq-def-r}, write
$r(u b,a) = \delta(u b,u b a) = |s|$, where $s$ is a shortest
distinguisher between $u b$ and $u b a$. Hence $s$ is some $t a$ with
$t\subword u b$ and $s\not\subword u b$. If $t\subword u$, then
$s\subword u a$ and $s\not\subword u$, so $s$ is a distinguisher
between $u$ and $u a$, thus $r(u ,a)=\delta(u,u
a)\leqby{\eqref{eq-delta-distinguisher-2}} |s|$. If, on the other
hand, $t\not\subword u$, then $t$ is a distinguisher between $u$ and
$ub$, so $r(u,b)=\delta(u,u b)\leqby{\eqref{eq-delta-distinguisher-2}}
|s|-1$. In both cases we deduce $r(u b,a)=|s|\geq
\min\bigl(r(u,b)+1,r(u,a)\bigr)$.

\noindent
($\leq$): Let $r(u,a)=|s|$, where $s$ is a shortest distinguisher
between $u$ and $u a$: necessarily $s$ is some $t a$ with $t\subword
u$ and $s\not\subword u$. Now $s=t a\not\subword u b$ since $a\neq b$.
With $s\subword u a\subword u b a$, we see that $s$ is a distinguisher
between $u b$ and $u b a$. We deduce $r(u b,a)\leq |s|=r(u,a)$.

If now we let $r(u,b)=|s|$, where $s$ is a shortest distinguisher
between $u$ and $u b$, then $s\not\subword u$, so $s a\not\subword u
b$. Since, on the other hand, $s a\subword u b a$, we see that $s a$ is
a distinguisher between $u b$ and $u b a$. We deduce $r(u b,a)\leq
|s|+1=r(u,b)+1$.

Joining the two deductions above yields $r(u b,a)\leq
\min\bigl(r(u,b)+1,r(u,a)\bigr)$ and concludes the proof.
\end{proof}

\begin{lemma}[Monotonicity of $r$ and $\ell$]
\label{lem-sld-mono}
For all $u,v,t\in A^*$
\begin{xalignat}{2}
\label{eq-sld-mono}
r(v,t)&\leq r(u v,t)
\:,
&
\ell(t,u)&\leq \ell(t,u v)
\:.
\end{xalignat}
\end{lemma}

{\let\endproof\doendproof
\begin{proof}
We prove the left-side inequality by induction on $|v|$ and then on
$|t|$, the right-side inequality being derived by mirroring. We
consider several cases:

(i) if $t=\epsilon$ then $r(v,t) = \infty = r(u v,t)$.

(ii) if $t$ is a letter not occurring in $v$, then
$r(v,t)\eqby{\eqref{eq-sld-0}}0\leq r(u v,t)$.

(iii) if $t=a$ is a letter occurring in $v$, we write $v$ under the
form $v=v' av''$ with $a$ not occurring in $v''$ and derive
\begin{gather*}
r(v,t)
\eqby{\eqref{eq-sld-u'}} 1+r(v',av'') \leqby{\text{i.h.}}
1+r(u v',av'') \eqby{\eqref{eq-sld-u'}} r(u v,t)
\:.
\end{gather*}

(iv) if $|t|>1$ then
\begin{gather}
\tag*{\qed}
r(v,t) \eqby{\eqref{eq-sld-t}}
\min_{a\in t}r(v,a) \leqby{\text{i.h.}}
\min_{a\in t}r(u v,a) \eqby{\eqref{eq-sld-t}}
r(u v,t)
\:.
\end{gather}
\end{proof}
}

Observe that the monotonicity of $r(v,t)$ with regards to $v$ only
holds when we extend $v$ to its left. And indeed $r(u,a)$
can be strictly larger than $r(u v,a)$, e.g., $r(a a,a) = 2 > r(a a
b,a) = 1$.

There is an upper bound on how much $r(v,t)$ and $\ell(t,u)$ may increase
when one add letters in $v$ or $u$  (at any position):
\begin{lemma}
For any $u,v\in A^*$ and $a,b\in A$
\begin{xalignat}{2}
\label{eq-mono-insert-1a}
r(u a v,b)&\leq 1+r(u v,b) \:,
&
\ell(b,u a v)&\leq 1+\ell(b,u v)  \:.
\end{xalignat}
\end{lemma}
\begin{proof}
We prove the first inequality by induction on the length of $v$, the second
inequality being derived by mirroring. For the base case, where
$v=\epsilon$, we have
\[
r(u a v,b) = r(u a,b)\leqby{\eqref{eq-new-r-rules*}} 1+r(u,b) = 1+r(u v,b)\:.
\]
For the inductive case, assume	$v=v' a'$ with $a'\in A$.
If $b=a'$, we have
\begin{align*}
r(u a v' b,b)\eqby{\eqref{eq-new-r-rules*}}
1+r(u a v',b)\leqby{\text{i.h.}}
2+r(u v',b)\eqby{\eqref{eq-new-r-rules*}}
1+r(u v' b,b) = 1+r(u v,b)
\:.
\end{align*}
If, on the other hand, $b\neq a'$, we have
\begin{align*}
&r(u  a v' a', b) \eqby{\eqref{eq-new-r-rules*}}
\min\left\{\! \begin{array}{c}
	1+r(u a v',b)\\ r(u a v',a')
	      \end{array}\!\right\}
\leqby{\text{i.h.}}
1+\min\left\{\! \begin{array}{c}
	1+r(u v',b)\\ r(u v',a')
	      \end{array}\!\right\}
\eqby{\eqref{eq-new-r-rules*}} 1+r(u v' a',b)\:.
\end{align*}
In both cases, we conclude $r(u a v,b)=r(u a v' a',b)\leq 1+r(u v'
a',b)=1+r(u v,b)$ as requested.
\end{proof}

\subsection{Relating $h$ and $\rho$}

Using side distances it is possible to compare $h$ and $\rho$.
\begin{theorem}
\label{thm-h-m}
$h(u) \geq 1+ \rho(u)$ for any word $u$.
\end{theorem}
\begin{proof}
For the empty word, one has $h(\epsilon)=1$ and $\rho(\epsilon)=0$.
We now assume that $u$ is non-empty:
by \Cref{prop-rho-charac}, there is a factorization $u=u_1 a u_2$ with
$a\in A$, such that
\begin{align*}
\rho(u) &= 1 + r(u_1,a) + \ell(a,u_2)
\eqby{\eqref{eq-sld-u'}} r(u_1a,a)+\ell(a,u_2)
\leqby{\eqref{eq-hu-delta-u1-u1a-u2-au2}}
h(u)-1
\:.
\end{align*}
and this concludes the proof.
\end{proof}

The difference between $h(u)$ and $\rho(u)$ can be larger than $1$,
see \Cref{rem-h-rho-gap} below. However, when $|A|\leq 2$ the
inequality becomes an equality, see \Cref{sec-binary}.

\subsection{Subword complexity and concatenation}
\label{ssec-subw-compl-concat}

While the subwords of $u v$ are obtained by concatenating the subwords
of $u$ and the subwords of $v$, there is no simple relation between
$h(u v)$ or $\rho(u v)$ on one hand, and $h(u)$, $h(v)$, $\rho(u)$ and
$\rho(v)$ on the other hand.

However, we can prove that $h$ and $\rho$ are monotonic and convex
with respect to concatenation.
We start with convexity.
\begin{theorem}[Convexity]
\label{thm-add-h-rho}
For all $u,v\in A^*$
\begin{gather}
\label{eq-add-rho}
\rho(u v) \leq \rho(u)+\rho(v)
\:,
\\
\shortintertext{and}
\label{eq-add-h-with-rho}
h(u v) \leq \max \bigl\{ h(u)+\rho(v), \: \rho(u)+h(v)\bigr\}
\:.
\end{gather}
\end{theorem}
Note that, thanks to \Cref{thm-h-m}, the second inequality entails
\begin{equation}
\label{eq-add-h}
h(u v)\leq h(u)+h(v) -1
\end{equation}
and is in fact stronger.

The proof of \Cref{thm-add-h-rho} relies on the following two lemmas:
\begin{lemma}
For any $u\in A^*$, for any $a\in A$
\begin{equation}
\label{eq-r-leq-rho}
r(u,a)\leq\rho(u)\:.
\end{equation}
\end{lemma}

{\let\endproof\doendproof
\begin{proof}
If $a$ does not appear in $u$, then $r(u,a)\eqby{\eqref{eq-sld-0}} 0$.
On the other hand,
if $u = u' a u''$ with $a \not\in \alphabet(u'' )$, then
\begin{align}
\notag
r(u,a) & \eqby{\eqref{eq-sld-u'}} r(u',a u'') + 1 \eqby{\eqref{eq-sld-t}}
\min \bigl\{r(u', b) ~|~ b\in\alphabet(a u'')\bigr\}+1
\\
\tag*{\qed}
       & \leq r(u',a)+1 \leq r(u',a)+\ell(a,u'') +1 \leqby{\eqref{eq-m-via-r&l}} \rho(u)
\:.
\end{align}
\end{proof}
}

\begin{lemma}
For any $u,v,t\in A^*$
\begin{xalignat}{2}
\label{eq-sld-extend3}
r(u v,t)&\leq \rho(u) + r(v,t)
\:,
&
\ell(t,u v)&\leq \rho(v) + \ell(t,u)
\:.
\end{xalignat}
\end{lemma}

\begin{proof}
We prove the first inequality by induction on $|v|$ and then on $|t|$,
the second inequality being derived by mirroring. We consider several
cases:
\\
(i) if $t=\epsilon$ then $r(u v,t) = \infty = r(u,t)$.
\\
(ii) if $t=a$ is a letter, there are three subcases: firstly, if
$v=\epsilon$, then $\rho(u v,t) = r(u,a) \leqby{\eqref{eq-add-rho}}
\rho(u)$; secondly, if $v = v' a$ ends with $a$ then
\[
r(u v,t)=r(u v' a,a)
\eqby{\eqref{eq-new-r-rules*}} 1 + r(u v',a)
\leqby{\text{i.h.}} \rho(u) + 1 + r(v',a)
\eqby{\eqref{eq-new-r-rules*}} \rho(u) + r(v,t)
\:;
\]
and, thirdly, if $v=v'b$ ends with $b\neq a$, we use
\begin{align*}
r(u v,t)=r(u v' b,a)
&\eqby{\eqref{eq-new-r-rules*}}
\min\bigl\{r(u v',b)+1, r(u v',a)\bigr\}
\\
&\leqby{\text{i.h.}}
\rho(u)+\min\bigl\{r(v',b)+1,r(v',a)\bigr\}
\eqby{\eqref{eq-new-r-rules*}} \rho(u)+r(v,t)
\:.
\end{align*}
(iii) if $|t|>1$ then
$
r(uv,t) \eqby{\eqref{eq-sld-t}}
\min_{a\in t}r(uv,a) \leqby{\text{i.h.}}
\rho(u)+\min_{a\in t}r(v,a) \eqby{\eqref{eq-sld-t}}
r(u v,t)
$.
\end{proof}

\begin{proof}[of \Cref{thm-add-h-rho}]
For \eqref{eq-add-rho} we assume that $u v\neq\epsilon$,
otherwise the claim is trivial.
Now by \Cref{prop-rho-charac}, $\rho(u v)$ is
$r(w_1,a)+\ell(a,w_2)+1$ for some letter $a$ and some
factorization $u v=w_1 a w_2$ of $u v$.
Let us assume w.l.o.g.\ (the other case
is symmetrical) that $w_1=u_1$ is a prefix of $u=u_1 a u_2$ so that
$w_2=u_2v$.
Now
\begin{gather*}
\rho(u)= r(u_1,a)+\ell(a,u_2v)+1
\leqby{\eqref{eq-sld-extend3}} r(u_1,a)+\ell(a,u_2)+1 + \rho(v)
\leqby{\eqref{eq-m-via-r&l}} \rho(u) + \rho(v)
\:.
\end{gather*}
We proceed similarly for  \eqref{eq-add-h-with-rho}: By \Cref{prop-hu-delta-u-u1au2}, $h(u v)$ is
$r(w_1,a)+\ell(a,w_2)+1$ for some letter $a$ and some factorization $u
v=w_1w_2$ of $u v$. If $w_1=u_1$ is a prefix of $u=u_1u_2$ so that
$w_2=u_2v$, then
\begin{gather*}
h(u v)= r(u_1,a) + \ell(a,u_2v) + 1
\leqby{\eqref{eq-sld-extend3}} r(u_1,a) + \ell(a,u_2) + \rho(v)+1
\leqby{\eqref{eq-hu-delta-u1-u1a-u2-au2}} h(u) + \rho(v) \:.
\end{gather*}
On the
other hand, if $w_2$ is a suffix of $v$, through the same reasoning we
get $h(u v)\leq \rho(u)+h(v)$.
\end{proof}

\begin{theorem}[Monotonicity]
\label{thm-mono-h}
For all $u,v\in A^*$
\begin{xalignat}{2}
\label{eq-mono-h}
h(u)&\leq h(u v)
\:,
&
h(v)&\leq h(u v)
\:,
\\
\label{eq-mono-rho}
\rho(u)&\leq \rho(u v)
\:,
&
\rho(v)&\leq \rho(u v)
\:.
\end{xalignat}
\end{theorem}
\begin{proof}
By \Cref{prop-hu-delta-u-u1au2}, $h(u)$ is $\delta(u_1u_2,u_1a u_2)+1$
for some letter $a\in A$ and some factorization $u=u_1 u_2$ of $u$. We
derive
\begin{gather*}
h(u)			\eqby{\eqref{eq-hu-delta-u1-u1a-u2-au2}}
r(u_1,a)+\ell(a,u_2)+1	\leqby{\eqref{eq-sld-mono}}
r(u_1,a)+\ell(a,u_2v)+1 \leqby{\eqref{eq-hu-delta-u1-u1a-u2-au2}}
h(u v)
\:.
\end{gather*}
For Eq.~\eqref{eq-mono-rho} we proceed similarly: if $u=\epsilon$ then
$\rho(u)=0$ and the claim follows. If $u\not=\epsilon$ then, by
\Cref{prop-rho-charac}, $\rho(u)$ is $\ell(v_1,a)+r(a,v_2)+1$ for a
factorization $u=v_1 a v_2$. We derive
\begin{gather*}
\rho(u)			    \eqby{\eqref{eq-m-via-r&l}}
\ell(v_1,a)+r(a,v_2)+1	    \leqby{\eqref{eq-sld-mono}}
\ell(v_1,a) + r(a, v_2 v)+1 \leqby{\eqref{eq-m-via-r&l}}
\rho(uv)
\:.
\end{gather*}
This accounts for the left halves of
(\ref{eq-mono-h},\ref{eq-mono-rho}). Their right halves follow by
mirroring.
\end{proof}

\subsection{Subword complexity and shuffle}

Beyond concatenations, we can consider shuffles of words.  Let us
write $u\shuffle v$ for the set of all words that can be obtained by
shuffling, or ``interleaving'', $u$ and $v$, i.e., the set of all
words $w = u_1 v_1 u_2 v_2 \cdots u_n v_n$ such that $u=u_1\cdots u_n$ and
$v=v_1\cdots v_n$~\cite{sakarovitch83}.	 This notion extends to the
shuffling of three and more words. In~\cite[Sect.~4]{HS-ipl2019}
Halfon and Schnoebelen proved the following bound:
\begin{gather}
\label{eq-from-HS-ipl2019}
w \in u_1\shuffle \cdots \shuffle u_m
\implies
h(w)\leq 1+\max_{a\in A} \: \bigl\{ |u_1|_a+\cdots +|u_m|_a \bigr\} \:,
\\
\shortintertext{with the corollary}
\label{eq-hu-leq-1+max-ua}
	     h(u)\leq 1+\max_{a\in \alphabet(u)}|u|_a \:.
\end{gather}

With our new tools, we can prove a bound mixing piecewise complexity and length.
\begin{proposition}
\label{prop-h-mono-shuffle}
For any $u,v\in A^*$
\begin{equation}
\label{eq-h-mono-shuffle}
w\in u\shuffle v \implies h(w)\leq h(u) + |v| \:.
\end{equation}
\end{proposition}

\begin{proof}
By induction over $|v|$. The claims holds obviously when $v=\epsilon$.

Assume now that $|v|=1$, i.e., that $v=a$ is a letter. Then $w$ is some $u_1 a u_2$ and
$h(w)=r(w_1,b)+\ell(b,w_2)+1$ for some factorizations $u=u_1u_2$,
$w=w_1w_2$ and some inserted letter $b$. W.l.o.g.\, we may assume that
$u_1a$ is a prefix of $w_1$, the other case being symmetrical. So let
us write $w=u_1 a w''_1 w_2$ and derive
\begin{align*}
h(w)
& = 1+r(u_1 a w''_1,b)+\ell(b,w_2)
\leqby{\eqref{eq-mono-insert-1a}}
 2+r(u_1 w''_1,b)+\ell(b,w_2)
\\
& \leq
1+h (u_1 w''_1 w_2)
=
1+h(u_1 u_2)=1+h(u)=h(u)+|v|
\:.
\end{align*}

Assume now that $|v|>1$ and write $v=v'a$ where $a$ is the last letter
of $v$.	 If $w\in u\shuffle (v' a)$ then $w\in w'\shuffle a$ for some
$w'\in u\shuffle v'$. We obtain
\begin{gather*}
h(w)\leqby{\text{i.h.}} h(w')+1\leqby{\text{i.h.}}h(u)+|v'|+1=h(u)+|v|
\:.
\end{gather*}
Finally all cases have been considered and the claim is proved.
\end{proof}

\begin{remark}
We thought for some time that the bound $h(u v)\leq h(u)+h(v)-1$ from
\eqref{eq-add-h} could be extended to a more general ``$h(w)\leq
h(u)+h(v)-1$ for all $w\in u\shuffle v$''. However this fails even
with a binary alphabet, e.g., taking $u=\mathtt{A A A B B B}$,
$v=\mathtt{A A A A B A A A A}$ and $w=\mathtt{A A A A A A A B A B A B A A B}$, where
$h(u)=4$, $h(v)=6$, $w\in u\shuffle v$ and $h(w)=12$.
\qed
\end{remark}

\section{Computing $h(u)$ and $\rho(u)$}
\label{sec-algo-h-and-rho}

With \Cref{prop-hu-delta-u-u1au2,prop-rho-charac}, computing $h$ and
$\rho$ reduces to computing $r$ and $\ell$.

The recursion in \Cref{eq-new-r-rules*} involves the prefixes
of $u$. We define the \emph{$r$-table of $u$} as the rectangular
matrix containing all the $r(u(0,i),a)$ for $i=0,1,\ldots,|u|$ and
$a\in A$. For simplicity we just write $r(i,a)$ for $r(u(0,i),a)$.
\begin{example}
\label{ex-w1}
Let $u=\mathtt{A B B A C C B C C A B A A B C}$ over
$A=\{\ta,\tb,\tc\}$. The $r$-table of $u$ is
\begin{center}
\scalebox{1.0}{%
\begin{tikzpicture}[auto,x=5mm,y=3mm,anchor=mid,baseline,node distance=1.3em]
{
\def\Wc{0.92}
\def\ixh{2.7}
\def\alphah{2.9}
\def\wlh{1.5}
\def\wlet{\ta, \tb, \tb, \ta, \tc, \tc, \tb, \tc, \tc,
\ta, \tb, \ta, \ta, \tb, \tc}
\def\riA{0, 1, 1, 1, 2, 1, 1, 1, 1, 1, 2, 2, 3, 4, 4, 3}
\def\riB{0, 0, 1, 2, 2, 1, 1, 2, 2, 2, 2, 3, 3, 3, 4, 3}
\def\riC{0, 0, 0, 0, 0, 1, 2, 2, 3, 4, 2, 2, 2, 2, 2, 3}
\node[left] at (-1.0,\ixh) {\tiny $i$};
\foreach \i in {0,...,15} \node at ({\i*\Wc},\ixh) {\tiny \i};
\node[left] at (-1.0,\wlh) {$w$};
\foreach \val [count=\i] in \wlet \node at ({(\i-0.5)*\Wc},\wlh) {$\val$};
\node[left] at (-1.0,0) {$r(i,\ta)$};
\foreach \val [count=\i] in \riA \node at ({(\i-1)*\Wc},0) {\val};
\node[left] at (-1.0,-1.3) {$r(i,\tb)$};
\foreach \val [count=\i] in \riB \node at ({(\i-1)*\Wc},-1.3) {\val};
\node[left] at (-1.0,-2.6) {$r(i,\tc)$};
\foreach \val [count=\i] in \riC \node at ({(\i-1)*\Wc},-2.6) {\val};
{\tikzstyle{every path}=[draw,thick,-]
\path (-0.8,\ixh) -- (-0.8,-2.9);
\path (-3.2,0.85) -- (15.5*\Wc-0.15,0.85);}
}
\end{tikzpicture}
}
\end{center}
With \Cref{eq-new-r-rules*}, it is easy to build such a table column
by colum.  Let us look at the colum $r(7,a)_{a\in A}$: since $u(0,7)
=\mathtt{A B B A C C B}$ ends with $\tb$, one has
\begin{align*}
r(7,\ta) &\eqby{\eqref{eq-new-r-rules*}}
	\min \bigl(r(6,\tb)+1,r(6,\ta)\bigr) = \min(2,1) = 1\:,
\\
r(7,\tb) &\eqby{\eqref{eq-new-r-rules*}}
	r(6,\tb)+1=2 \:,
\\
r(7,\tc) &\eqby{\eqref{eq-new-r-rules*}}
	\min \bigl(r(6,\tb)+1,r(6,\tc)\bigr) = \min(2,2) = 2\:.
\end{align*}

\end{example}

\begin{corollary}
\label{coro-algo-h}
$h(u)$ can be computed in quadratic time
$O(|A|\cdot |u|)$, hence in linear time when $A$ is fixed.
\end{corollary}
\begin{proof}
It is clear that \Cref{eq-new-r-rules*} directly leads to a
$O(|A|\cdot |u|)$ algorithm for computing the $r$-table and,
symmetrically, the $\ell$-table for $u$.  One then finds $\max_{a\in
A}\max_{i=0,\ldots,|u|} r(u(0,i),a)+\ell(a,u(i,|u|))+1$ by looping
over the two tables. As stated in \Cref{prop-hu-delta-u-u1au2}, this
gives $h(u)$.
\end{proof}

In order to compute $\rho(u)$, we use the $r$- and $\ell$-vectors.
\begin{definition}[\protect{$r$-vector, $\ell$-vector, \cite[p.~73]{simon72}}]
\label{def-r-l-vectors}
The $r$-vector of $u=a_1\cdots a_m$ is $\tuple{r_1,\ldots,r_m}$,
defined with $r_i=r(a_1\cdots a_{i-1}, a_i)$ for all $i=1,\ldots,m$.\\
The $\ell$-vector of $u$ is $\tuple{\ell_1,\ldots,\ell_m}$, defined
symmetrically with $\ell_i=\ell(a_i,a_{i+1}\cdots a_m)$ for all $i=1,\ldots,m$.\\
\end{definition}
With these definitions, the following equality is just a rewording of \Cref{prop-rho-charac}:
\begin{equation}
\label{eq-rho-by-rl}
\rho(u)=1+\max_{i=1,\ldots,m}(r_i+\ell_i) \:.
\end{equation}
\begin{example}
Let us continue with $u=\mathtt{A B B A C C B C C A B A A B C}$. Its $r$- and
$\ell$-vectors are, respectively:
\[
\begin{array}{r c c c c c c c c c c c c c c c c c}
		      & &\ta\:&\tb\:&\tb\:&\ta\:&\tc\:&\tc\:&\tb\:&\tc\:&\tc\:&\ta\:&\tb\:&\ta\:&\ta\:&\tb\:&\tc\:&
\\
\text{$r$-vector:}    & \langle & 0,& 0,& 1,& 1,& 0,& 1,& 1,& 2,& 3,& 1,& 2,& 2,& 3,& 3,& 2&\rangle
\\
\text{$\ell$-vector:} & \langle & 3,& 4,& 3,& 2,& 4,& 3,& 2,& 2,& 1,& 2,& 1,& 1,& 0,& 0,& 0&\rangle
\end{array}
\]
By summing the two vectors, looking for a maximum value, and adding 1,
we quickly obtain $\rho(u) = 1+ \max_{i=1,\ldots,|u|} (r_i+\ell_i)  =
5$.
\qed
\end{example}

\begin{remark}
The \emph{attribute} of $u$ defined in
\cite[\textsection~3]{fleischer2018} is exactly the juxtaposition of
Simon's $r$- and $\ell$-vectors with all values shifted by 1.
\qed
\end{remark}

One could extract the $r$- and $\ell$-vectors from the $r$- and
$\ell$-tables but there is a faster way. Indeed, writing
$\tuple{r_1,r_2,\ldots, r_m}$ for $u$'s $r$-vector, one has
\begin{equation}
\label{eq-algo-ri}
r_i = \begin{cases}
	1+\min(r_j,r_{j+1},\ldots,r_{i-1}) &\text{if $u(i)$ last
	occurred as $u(j)$,}
\\
	0 & \text{if $u(i)$ did not occur previously in $u$.}
     \end{cases}
\end{equation}
From this, one may derive an algorithm that builds the $r$-vector by a
one-pass traversal of $u$.

\Cref{fig-algo-rvector} provides explicit
code, extracted from the algorithm in~\cite{barker2020} that computes
the canonical representative of $u$ modulo $\sim_k$. While traversing
$u$, it maintains a table $\texttt{locc}:A\to\Nat$ of positions of
last occurrences of each letter, as well as a stack $\texttt{L}$ of
positions of interest.
We refer to \cite{barker2020} for the
correctness proof. The running time is $O(|A|+|u|)$ since there is a linear
number of insertions in the stack $\texttt{L}$ and all the positions
read from $\texttt{L}$ are removed except the last read.
\begin{figure}[htp]
\centering
\begin{tabular}{c}
\lstset{language=python,texcl=true,basicstyle=\tt,commentstyle=\it,mathescape=true}
\begin{lstlisting}
function r-vector(u,A):
    # initialize {\rm\texttt{locc}} and {\rm\texttt{L}}
    for a in A: locc[a]=0;
    L.push(0)
    # fill {\rm\texttt{r}} and maintain {\rm\texttt{locc}}:
    for i from 1 to |u|:
	a = u[i];
	while (head(L) >= locc[a]):
	    j = L.pop()
	r[i] = 1+r[j] if j>0, else 0
	L.push(j)
	L.push(i)
	locc[a]=i
    return r
\end{lstlisting}
\end{tabular}
\caption{Linear-time algorithm for $r$-vector, after~\cite{barker2020}.}
\label{fig-algo-rvector}
\end{figure}
With a mirror algorithm, the $\ell$-vector is also computed in linear time.
\begin{corollary}
\label{coro-algo-rho}
$\rho(u)$ can be computed in linear time $O(|A| + |u|)$.
\end{corollary}

\section{Piecewise complexity of binary words}
\label{sec-binary}

The case where $|A|=2$ is special and enjoys some special properties.

To begin with, the inequality in \Cref{thm-h-m} becomes an equality
with binary words.
\begin{theorem}
\label{thm-h-m-2letter}
Assume	$|A|= 2$.  Then	 $h(u) = \rho(u) + 1$ for any $u\in A^*$.
\end{theorem}
\begin{proof}
In view of \Cref{thm-h-m}, there only remains to prove $h(u)\leq \rho(u)+1$.
Let us write $A=\{a,b\}$ and assume $u\in A^*$.
Using \eqref{eq-hu-delta-u1-u1a-u2-au2} we know that
$h(u)=1+r(u_1,x)+\ell(x,u_2)$ for some letter $x\in A$ and
factorization
$u=u_1u_2$. W.l.o.g.\ we can assume $x=a$ and consider	several
possibilities
for $u_1$ and $u_2$.
\begin{description}

\item[Case 1:] $u_1$ ends with $a$, i.e., $u_1 = u'_1 a$. Then
\begin{equation}
\begin{aligned}
\notag
 h(u) &= 1+r(u'_1 a,a) +\ell(a,u_2) \eqby{\eqref{eq-sld-u'}} 2+r(u'_1,a) + \ell(a,u_2) \\
	 &\leqby{\eqref{eq-m-via-r&l}} 1+\rho(u'_1a u_2)=1+\rho(u) \:.
\end{aligned}
\end{equation}
\item[Case 2:] $u_1$ ends with $b$ and $u_2$ starts with $b$, i.e.,
  $u_1 = u'_1 b$, $u_2 = bu'_2$, and $u= u'_1 b b u'_2$. Then
\begin{equation}
\begin{aligned}
\notag
 h(u) &= 1+r(u_1,a) +\ell(a,u_2) = 1+r(u'_1 b,a) +\ell(a,bu'_2)
\\
      &\leqby{\eqref{eq-new-r-rules*}} 1 + (1+r(u'_1,b)) + (1+\ell(b,u'_2))
      = 3 + r(u'_1,b) + \ell(b,u'_2)
\\
      &\eqby{\eqref{eq-sld-u'}} 2+r(u'_1b,b) + \ell(b,u'_2) \leqby{\eqref{eq-m-via-r&l}} 1+\rho(u'_1b b u'_2)=1+\rho(u) \:.
\end{aligned}
\end{equation}
\item[Case 3:] $u_1$ ends with $b$ and $u_2$ is	 empty, i.e.,
  $u_1 = u'_1 b$ and $u=u_1$. Then
\begin{equation}
\begin{aligned}
\notag
 h(u) &= 1+r(u_1,a) +\ell(a,u_2) = 1+r(u'_1 b,a) +\ell(a,\epsilon)
\\
	&\eqby{\eqref{eq-sld-0}} 1+r(u'_1 b,a) +\ell(b,\epsilon) \leqby{\eqref{eq-new-r-rules*}} 1 + 1+r(u'_1,b) + \ell(b,\epsilon)
\\
      &\leqby{\eqref{eq-m-via-r&l}} 1+\rho(u'_1 b u'_2)=1+\rho(u) \:.
\end{aligned}
\end{equation}
\item[Remaining Cases:] If $u_2$ starts with $a$, we use symmetry and
  reason as in Case 1.	If $u_2$ starts with $b$ and $u_1$ is empty,
  we reason as in Case 2.  If both $u_1$ and $u_2$ are empty, we have
  $u=\epsilon$, $h(\epsilon)=1$ and $\rho(\epsilon)=0$.
\end{description}
All cases have been dealt with and the proof is complete.
\end{proof}

\begin{remark}
\label{rem-h-rho-gap}
\Cref{thm-h-m-2letter} cannot be generalised to words using $3$ or more
different letters.  For example, with $A=\{\ta,\tb,\tc\}$ and
$u=\mathtt{C A A C B A B A}$,
one has $h(u)=5$ and $\rho(u)=3$.  Larger gaps are possible:
$u=\mathtt{C B C B C B C B B C A B B A B A B A A A}$ has $h(u)=10$ and $\rho(u)=6$.
\qed
\end{remark}

When dealing with binary words, the alternation of letters is fixed
and it is  convenient  to describe a word by just its first letter and
the lengths of its blocks.
In this spirit, we would write a binary word that starts with
$\ta$ under the form
$u = \ta^{n_1}\tb^{n_2}\ta^{n_3}\tb^{n_4}\cdots (\ta/\tb)^{n_k}$ where
the last letter is $\ta$ or $\tb$ depending on the parity of $k$ and
where $n_i>0$ for all $i=1,\ldots,k$. It turns out that we can compute
$h(u)$ and $\rho(u)$ efficiently when the number $k$ of blocks is
notably smaller than the length of $u$. We now explain the underlying
technique as a succession of cumulative steps for which the rather
technical proofs have been relegated to the Appendix.

In a word of the form
$u = \ta^{n_1}\tb^{n_2}\ta^{n_3}\tb^{n_4}\cdots (\ta/\tb)^{n_k}$ an
\emph{isolated} (occurrence of a) letter is any block
with $n_i=1$. Our first result gives a direct formula for
the piecewise complexity of binary words with no isolated letters
except perhaps at the extremities.
\begin{lemma}[See Appendix]
\label{lem-h-binary-ni>=2}
Let $u=\ta^{n_1}\tb^{n_2}\ta^{n_3}\cdots(\ta/\tb)^{n_k}$ with $k\geq
2$. If $n_i\geq 2$ for all $i=2,3,\ldots,k-1$ and $n_1,n_k\geq 1$ then
\begin{equation}
\label{eq-rho-altab}
 \rho(u) = k + \max\biggl(\begin{array}{c}
				n_1{+}1,\: n_k{+}1,
				\\
				n_2,n_3,\ldots,n_{k-1}
		      \end{array}\biggr)-3.
\end{equation}
\end{lemma}
For example, \Cref{lem-h-binary-ni>=2} tells us that, for $u =
\ta^{34}\tb^{23}\ta^{18}\tb$, one has $\rho(u) = 4 + \max(35,23,18,2)
- 3 = 36$, (and $h(u) = 37$ thanks to \Cref{thm-h-m-2letter}).
It also validates the claim $h(\ta\tb\tb\ta)$=3 made in the introduction
of this article.
\\

There remains to account for isolated letters, where $n_i=1$. When
they appear in pairs, we invoke the following lemma.
\begin{lemma}[See Appendix]
\label{lem-h-binary-u11v}
Assume $A=\{\ta,\tb\}$ and  $u,v\in A^*$.
If $u$ does not end with $\ta$ and $v$ does not start with $\tb$
then $\rho(u\ta\tb v)=1+\rho(u v)$.
\end{lemma}

Finally, for isolated letters that are surrounded by multiple letters,
we use another lemma.
\begin{lemma}[See Appendix]
\label{lem-h-binary-un1mv}
Assume $A=\{\ta,\tb\}$ and $u,v\in A^*$.  Then $\rho(u \ta\ta\tb\ta\ta
v)=\rho(u\ta\ta\ta\tb\tb\tb \tilde{v})$ where $\tilde{v}$ is the word
obtained from $v$ by replacing any $\ta$ by a $\tb$ and vice versa.
\end{lemma}

With these three lemmas we have a complete method for binary words.
\begin{theorem}
For binary words $u\in \{\ta,\tb\}^*$, $h(u)$ and $\rho(u)$ can be
computed in time $O(|u|)$, even if $u$ is given in exponential
notation $u=a_1^{n_1}a_2^{n_2} \cdots a_k^{n_k}$ with input size
$|u|\eqdef \sum_{i=1}^k(1+\log_2 n_i)$.
\end{theorem}
Let us explain the algorithm on an example.  Given
$u=\ta^1\tb^5\ta^1\tb^1 \ta^1\tb^1\ta^1\tb^7\ta^1\tb^4$, we first
remove the pairs of isolated letters. This
yields $u'=\ta^1\tb^5\ta^1\tb^7\ta^1\tb^4$ with $\rho(u)=2+\rho(u')$.
Then we remove the first isolated letter surrounded by larger blocks,
obtain $u''=\ta^1\tb^6\ta^8\tb^1\ta^4$, remove the other and obtain
$u'''=\ta^1\tb^6\ta^9\tb^5$ with $\rho(u''')=\rho(u'')=\rho(u')$. Here
\Cref{lem-h-binary-ni>=2} applies and yields
$\rho(u''')=4+\max(1+1,6,9,5+1)-3=10$. Thus $\rho(u)=12$ and
$h(u)=13$.
Finally,
the above algorithm can work directly on a list of pairs
$(a_1,n_1), (a_2,n_2), \ldots, (a_k,n_k)$.
\\

An interesting application of these results is that we can derive an
exact value for the maximum length of a binary word having a given piecewise complexity.
\begin{theorem}
\label{thm-maxlenu-binary}
For any $u$ binary word the following holds:
\begin{equation}
\label{eq-maxlenu-binary}
|u|\leq \Bigl\lfloor\frac{h(u)^2}{4}\Bigr\rfloor+h(u)-1\:.
\end{equation}
Furthermore, the bound is tight.
\end{theorem}
\begin{proof}
We first argue that the claimed bound is reached. If $h(u)=2g$ is even, the bound is reached
by taking $u$ made of $g+1$ alternating blocks of length $g+1$ each
except for the two extreme blocks of length $g$. E.g., $h(u)=6$ is
obtained with $u=\ta^3\tb^4\ta^4\tb^3$ having maximum length 14.
For odd $h(u)=2g-1$ the claimed bound is reached with $g$ blocks of length
$g+1$ (except for extreme blocks). E.g.,  $h(u)=7$ is obtained
with $u=\ta^4\tb^5\ta^5\tb^4$ having maximum length 18.

We let the reader prove that \Cref{eq-maxlenu-binary} holds. This is
direct from \Cref{lem-h-binary-ni>=2} when $u$ has no isolated
letters. And if $u$ has isolated letters, \Cref{lem-h-binary-u11v} or
\Cref{lem-h-binary-un1mv} shows how to build a longer word with same $h(u)$.
\end{proof}

\section{Conclusion}
\label{sec-concl}

In this article, we focused on the piecewise complexity of individual
words, as captured by the piecewise height $h(u)$ and the minimality
index $\rho(u)$, a new measure suggested by~\cite{simon72} and that we
introduce here.

These measures admit various characterisations, including
\Cref{prop-hu-delta-u-u1au2,prop-rho-charac}, that can be leveraged
into efficient algorithms running in quadratic time $O(|A|\cdot|u|)$
for $h(u)$ and linear time $O(|A|+|u|)$ for $\rho(u)$. Our analysis
further allows to establish monotonicity and convexity properties for
$h$ and $\rho$
and to relate $h$ and $\rho$.  \\

This first foray into the piecewise complexity of individual words
raise a wide range of new questions. Let us mention just two:
\begin{enumerate}
\item
The minimality index was found to be very useful in the analysis of
binary words but its usefulness could certainly be improved. For
example, is there a way to define $\rho(L)$ for a PT-language $L$?
Can we relate $\rho(u)$ to some more usual measures of logical complexity?
\item
For
fixed alphabet size $k=|A|$, what are the longest words having
piecewise complexity at most $n$? A lower bound can be found
in~\cite{KS-lmcs2019} where is given a family of words
$(U_{k,\lambda})_{k,\lambda\in\Nat}$ having length $(\lambda+1)^k-1$
and piecewise complexity $k\lambda+1$.	An upper bound is given
in~\cite[Prop.~3.8]{KS-lmcs2019}:
\begin{equation}
	|u|\leq \biggl(\frac{h(u)-1}{|A|}+2\biggr)^{|A|}-1 \:.
\end{equation}
Here the lower and the upper bounds are similar but do not match
exactly.  For $|A|=2$ we can give the exact maximum length of $u$ as a
function of $h(u)$, see \Cref{thm-maxlenu-binary}, but for $|A|>2$ the problem is open.
\end{enumerate}

\section*{Acknowledgements}

This article is an updated and extended presentation of results that
first appeared in~\cite{PSVV-sofsem2024}, where they are supplemented
by algorithms, obtained in collaborations with M.~Praveen and
J.~Veron, for computing $h(u^n)$ and $\rho(v^n)$ efficiently. It
profited from remarks and suggestions by our two co-authors, as well
as further suggestions from anonymous reviewers.



\bibliographystyle{alpha}
\bibliography{subwords}

\appendix
\section{Proofs for \Cref{sec-binary}}

Let us fix $u = \ta^{n_1}\tb^{n_2}\ta^{n_3}\tb^{n_4}\cdots
(\ta/\tb)^{n_k}$ with $n_i\geq 1$ for all $i=1,\ldots,k$ and write
$\lambda_i \eqdef n_1+\cdots+n_i$ for the cumulative length of the
first $i$ blocks, so that $|u|=\lambda_k$. A position
$p\in\{1,2,\ldots,|u|\}$ has a unique decomposition of the form
$p=\lambda_i+d$ with $0\leq i<k$ and $1\leq d\leq n_i$.

\subsection{Proof of \Cref{lem-h-binary-ni>=2}}

We are now ready to consider the $r$-vector $\tuple{r_1, r_2, \ldots,
r_{|u|}}$ of $u$ as given in \Cref{def-r-l-vectors}.
\begin{lemma}
Assume $u$ satisfies $n_i\geq 2$ for $i=2,\ldots,k-1$.	Then, for all
$0\leq i<k$ and $1\leq x\leq n_{i+1}$,
\begin{equation}
\label{eq-rix}
r_{\lambda_i+x} = \begin{cases} i+x-1 &\text{if $i=0$,}\\
			     i+x-2 &\text{if $i>0$.}
		\end{cases}
\end{equation}
\end{lemma}

\begin{proof}
By induction on $p=\lambda_i+x$. When $x>1$ the induction step is
trivial since the previous occurrence of $u(p)$ is at $p-1$, so
\Cref{eq-algo-ri} directly gives $r_p=1+r_{p-1}$.  Assume now $x=1$ so
that $p=\lambda_i+1$ is the first position in block $i+1$.
If $i=0$ we are at the first letter of $u$ so that $r_{\lambda_0+1}=r_1=0$ which
agrees with \eqref{eq-rix}.
If $i=1$ we are at the first occurrence of $\tb$ in
$u$ so that $r_{\lambda_1+1} \eqby{\eqref{eq-algo-ri}} 0$, agreeing with \eqref{eq-rix}.
If now $i\geq 1$ then ---assuming that $i$ is, say, even in order to
simplify notations--- $r_{\lambda_i+1}$ is
$r(\ta^{n_1}\cdots\tb^{n_i},\ta)$ and the previous occurrence of
$u(p)$ is at $\lambda_{i-2}+n_{i-1}$. Thus
\begin{align*}
r_{\lambda_i+1} \eqby{\eqref{eq-algo-ri}} &
		1+\min(r_{\lambda_{i-2}+n_{i-1}}, r_{\lambda_{i-1}+1}, \ldots, r_{\lambda_{i-1}+n_i})
\\
	      =\: &
		 1+\min(r_{\lambda_{i-2}+n_{i-1}},r_{\lambda_{i-1}+1})
\\
	      \eqby{\text{i.h.}}\:&
	      1+\min(i-2+n_{i-1}-(1\text{ or }2),i-1+1-2)
\\
\shortintertext{and, since $n_{i-1}- (1\text{ or }2)$ is always $\geq 0$,}
	       =\: & i+x-2\:.
\end{align*}
\end{proof}

By symmetry we obtain the following expression for the $\ell$-vector, on the
condition that $n_2,\ldots,n_{k-1}\geq 2$:
\begin{equation}
\label{eq-lix}
\ell_{\lambda_i+x} = \begin{cases} k-i+n_i-x-1 &\text{if $i=k-1$,}\\
			     k-i+n_i-x-2 &\text{if $i<k-1$.}
		\end{cases}
\end{equation}
Summing the two vectors, looking for a maximum value, and adding $1$,
we obtain $ \rho(u) = k + \max\bigl( n_1+1, n_k+1, n_2, n_3, \ldots,
n_{k-1}\bigr)-3 $ as claimed in \Cref{lem-h-binary-ni>=2}: In each
block the sum $r_{\lambda_{i-1}+d}+\ell_{\lambda_{i-1}+d}$ is
$k+n_i-2$ except in the first and last blocks where it is $k+n_1-1$
and $k+n_k-1$, respectively.

\subsection{Proof of \Cref{lem-h-binary-u11v}}

Recall that the claim is $\rho(u\ta\tb v)=1+\rho(u v)$ when the $\ta$
and the $\tb$ are isolated, i.e., when $u$ does not end with $\ta$ and
$v$ does not start with $\tb$.

Let us assume that both $u$ and $v$ are not empty and write them
under the form $u=u'\tb$ and $v=\ta v'$.  Write $w$ for
$u'\tb\ta\tb\ta v'$ and $w'$ for $u'\tb\ta v'$. Further define
$\Lambda=|w|$, the total length, and $K=|u\ta|$, the critical position
inside $w$.  Finally write $\tuple{r_1,\ldots,r_\Lambda}$ and
$\tuple{r'_1,\ldots,r'_{\Lambda-2}}$ for the $r$-vectors of $w$ and $w'$,
respectively.  The following diagram will help follow our
reasoning (it assumes that if $\tb$ occurs in $v'$ then its first
occurrence is at position $\delta$).
\[
\begin{array}{r c c c c c c c c c c c c c c c c c}
w{:}	      & &u(1)\: &u(2)\:&\cdots\:&\tb\:&\ta\:&\tb\:&\ta\:&\cdots\:&\tb\:&\cdots
\\
\text{$r$-vector:}   & \langle & r_1,& r_2,& \cdots,& r_{K-1},& r_K,& r_{K+1},& r_{K+2},& \cdots,& r_{K+2+\delta},& \cdots
\\[.5em]
w'{:}	      & &u(1)\: &u(2)\:&\cdots\:&\tb\:&	  &   &\ta\:&\cdots\:&\tb\:&\cdots
\\
\text{$r$-vector:}   & \langle & r'_1,& r'_2,& \cdots,& r'_{K-1},& & & r'_K,& \cdots, & r'_{K+\delta},& \cdots
\\
\end{array}
\]
Obviously $r_i=r'_i$ when $1\leq i \leq K$ since the values at position
$i$ in some $r$-vector only depend on the first $i$ letters of a word.
\begin{lemma} $r_{i+2}=1+r'_{i}$ when $K\leq i\leq \Lambda-2$.
\end{lemma}
\begin{proof}
By induction on $i$.  We use \Cref{eq-algo-ri} repeatedly, and one of
its consequences in particular: values in a $r$-vector cannot increase
by more than 1, i.e., $r_{i+1}\leq r_i+1$.

Now let us assume $i=K$.  Then $r_{K+2} = \min(1+r_K,1+r_{K+1})$ so,
using $r_{K+1} = \min(1+r_{K-1},1+r_K)$, we have $r_{K+2} =
\min(1+r_K,2+r_{K-1},2+r_K)$ hence
\begin{equation}
\label{eq-rK2-1+rK}
		     r_{K+2} = 1 + r_k = 1 + r'_K
\end{equation}
since $r_K\leq 1+r_{K-1}$. This proves the claim for $i=K$.

For $i=K+1, K+2,\ldots$ the induction is easy as long as
$w'(K+1),w'(K+2),\ldots$ are the same letter $\ta$: the values in both
$r$-vectors increase by 1 at each step.

When/if we reach $w'(K+\delta)=\tb$ for the first time we use
\[
r'_{K+\delta} = \min(1+r'_{K-1},1+r'_K) = \min(1+r_{K-1},1+r_K)=r_{K+1}\:.
\]
On the other hand, and since $r_{K+1}\leq 1+r_K$,
\[
r_{K+2+\delta}=\min(1+r_{K+2},1+r_{K+1})\eqby{\eqref{eq-rK2-1+rK}} \min(2+r_K,1+r_{K+1})=1+r_{K+1} \:.
\]
This
proves the claim for $i=K+\delta$.

Beyond $K+\delta$ the values in both $r$-vectors only depend on values
past the critical point $K$, where $r_{i+2}=1+r'_i$ holds by the
induction hypothesis, so they keep evolving in lock-step maintaining
the +1 difference.
\end{proof}

By symmetry we deduce similar relations between the $\ell$-vectors
$\tuple{\ell_1,\ldots,\ell_\Lambda}$ of
$w$ and $\tuple{\ell'_1,\ldots,\ell'_{\Lambda-2}}$ of $w'$, namely
$\ell_i=1+\ell'_i$ for $i=1,\ldots,K-1$
and $\ell_i=\ell'_{i-2}$ for $i=K+2,\ldots,\Lambda$.

Now $\rho(w')$, i.e., $1+\max_i(r'_i+\ell'_i)$ by \Cref{eq-rho-by-rl},
is $1+r'_{i_0}+\ell'_{i_0}$ for some position $i_0$ and we deduce
$r_{i_0}+\ell_{i_0}=r'_{i_0}+(1+\ell'_{i_0})$ if $i_0<K$, or
$r_{i_0+2}+\ell_{i_0+2}=(1+r'_{i_0})+\ell'_{i_0}$ if $i_0\geq K$.  Both
cases yield $\rho(w)\geq 1+\rho(w')$.

In order to show the reverse inequality, i.e., $\rho(w)\leq 1+
\rho(w')$, we need to prove that $\rho(w) = 1+\max_i(r_i+\ell_i)$ can
be reached with $i\not\in\{K,K+1\}$. For $i=K$, recall first that
$r_{K+2}\eqby{\eqref{eq-rK2-1+rK}}1+r_K$. Since $\ell_k\leq
1+\ell_{K+2}$ (by symmetry from already noted $r_{K+1}\leq 1+r_{K-1}$)
we see that $r_K+\ell_K\leq r_{K+2}+\ell_{K+2}$.  By symmetry we
deduce $r_{K+1}+\ell_{K+1}\leq r_{K-1}+\ell_{K-1}$. Finally we see
that $\max_i(r_i+\ell_i)$ can be reached with $i<K$ or $i\geq K+2$,
proving $\rho(w)\leq 1+ \rho(w')$ and establishing the claim for
non-empty $u,v$.
\\

The cases where one of $u$ and $v$ is $\epsilon$ is handled in a
similar way (actually it is a bit simpler) and we omit its detailed
analysis. Finally the claim is obviously true when $u=v=\epsilon$: we
noted that $\rho(\epsilon) = 0$ and one readily checks that
$\rho(\ta\tb) = 1$.

\subsection{Proof of \Cref{lem-h-binary-un1mv}}

Write $w$ for $u\ta\ta\tb\ta\ta v$ and $\tuple{r_1,\ldots,r_{|w|}}$
and $\tuple{\ell_1,\ldots}$ for its $r$- and $\ell$- vectors.  Write
$w'$ for $u\ta\ta\ta\tb\tb\tb \widetilde{v}$ and use $r'_i$ and
$\ell'_i$ for its vectors.  The claim we have to prove is
$\rho(w)=\rho(w')$. The critical position is $K=|u|+2$ and the
following diagram show how the two $r$-vectors relate:

\[
\begin{array}{r c c c c c c c c c c c c c c c c c}
w{:} &	     &u(1)\:&\cdots\:&\ta\:&\ta\:& &\tb\:&\ta\:&\ta\:&v(1)\:&\cdots\:&\tb\:&\cdots
\\
     &\langle& r_1, &\cdots,& r_{K-2},& r_{K-1}, & & r_{K},& r_{K+1},& r_{K+2},& \cdots&\cdots&r_{K+\delta},&\cdots
\\[.5em]
w'{:}&	     &u(1)\:&\cdots\:&\ta\:&\ta\: &\ta\:&\tb\:&\tb\:&\tb\:&\widetilde{v(1)}\:&\cdots&\ta\:&\cdots
\\
     &\langle&r'_1, & \cdots,& r'_{K-2},& r'_{K-1}, & r'_K,& r'_{K+1},& r'_{K+2}, & r'_{K+3},& \cdots &\cdots&r'_{K+\delta+1},&\cdots
\\
\end{array}
\]
Obviously $r'_i=r_i$ when $1\leq i < K$.
\begin{lemma}
$r'_{i+1}=r_i$ when
$K\leq i\leq|w|$.
\end{lemma}
\begin{proof}
By induction on $i$. We start with $i=K$: if $\tb$ does not occur in
$u$ then $r_K=0$ and $r'_{K+1}=0$, while if $\tb$'s last occurrence in
$u$ is at position $p$ then $r_K=1+\min\{r_j~|~j=p,\ldots,K-1\}$ while
$r'_{K+1}=1+\min\{r'_j~|~j=p,\ldots,K\}$. The two values coincide
since $r'_{K}>r'_{K-1}$ and since $r_j=r'_j$ for $j<K$.

The case for $i=K+1$ is easier: $r'_{K+2}=1+r'_{K+1}$ while
$r_{K+1}=1+\min(r_{K-1},r_K)=1+r_K$ since
$r_K\leqby{\eqref{eq-algo-ri}}1+r_{K-2}\eqby{\eqref{eq-algo-ri}}r_{K-1}$. From
them the values increment in lock-step in both vectors until we reach
the first occurrence of $\tb$ in $v$ at position $i=K+\delta$.

At that position $r_{K+\delta}=1+\min(r_K,r_{K+1})=1+r_K$ since
$r_{K+1}=1+r_K$ as seen above.
On the other hand $r'_{K+\delta+1}=1+\min(r'_K,r'_{K+1})=1+r'_{K+1}$
since
$r'_K= 2+r'_{K-2}$ while $r'_{K+1}\leq 1+r'_{K-2}$.  Finally
\[
r_{K+\delta}=1+r_K\eqby{\text{i.h.}}1+r'_{K+1}=r'_{K+\delta+1}\:.
\]
From then on, the values in both vectors only depend on previous
values past the critical point, and since (modulo the swapping of
$\ta$ with $\tb$) the letters agree in both suffixes, the induction
goes without any difficulty.
\end{proof}
A symmetrical results holds for the $\ell$-vectors, from which we conclude
$ \rho(w)\leq \rho(w')$ since all the sums $1+r_i+\ell_i$ on $w$'s
side also exist on $w'$'s side.

Then, in order to conclude the proof of \Cref{lem-h-binary-un1mv}, we
only need to show that $\rho(w')\leq \rho(w)$, i.e., that
$1+r'_K+\ell'_K$, the only sum that is not present on $w$'s side, can
also be reached or bested at some position $i\neq K$ in $w'$. But this
is easy to show, by observing, e.g., that $\ell_{K-1}=1+\ell'_K$ and
that $r'_{K}\leq 1+r'_{K-1}$ so that $r'_K + \ell'_K \leq r'_{K-1} +
\ell'_{K-1}$.

\end{document}